\documentclass[letterpaper, 10pt, conference]{ieeeconf}
\IEEEoverridecommandlockouts
\usepackage{algorithmic}
\usepackage{amsmath,amssymb,amsfonts}
\usepackage{booktabs}
\usepackage{cite}
\usepackage{comment}
\usepackage{enumerate}  
\usepackage[acronym]{glossaries}
\usepackage{graphicx}
\usepackage{mathtools}
\usepackage{textcomp}
\usepackage{tikz}
\usepackage{url}
\usepackage{xcolor}
\usetikzlibrary{positioning,calc,arrows.meta,fit,backgrounds}

\DeclareMathOperator*{\argmin}{arg\,min}

\newtheorem{remark}{Remark}

\definecolor{lightblue}{rgb}{0.60784,0.76078,0.90196}
\definecolor{darkblue}{rgb}{0.26667,0.44706,0.76863}
\definecolor{lightgreen}{rgb}{0.66275,0.81569,0.55686}
\definecolor{darkgreen}{rgb}{0.43922,0.67843,0.27843}
\definecolor{orange}{rgb}{0.92941,0.49020,0.19216}
\definecolor{yellow}{rgb}{1.00000,0.75294,0.00000}
\definecolor{grey}{rgb}{0.64706,0.64706,0.64706}
\definecolor{purple}{rgb}{0.51373,0.23529,0.04706}

\newacronym{abk:arz}{ARZ}{Aw--Rascle and Zhang}
\newacronym{abk:cav}{CAV}{Connected Autonomous Vehicle}
\newacronym{abk:cavs}{CAVs}{Connected Autonomous Vehicles}
\newacronym{abk:ge}{GE}{Generalization Error}
\newacronym{abk:lwr}{LWR}{Lighthill--Whitham--Richards}
\newacronym{abk:ml}{ML}{Machine Learning}
\newacronym{abk:pde}{PDE}{Partial Differential Equation}
\newacronym{abk:pdes}{PDEs}{Partial Differential Equations}
\newacronym{abk:pil}{PIL}{Phyics-Informed Learning}

\newacronym{abk:sumo}{SUMO}{Simulation of Urban MObility}
\newcommand{\Vmax}{V_\mathrm{max}}
\newcommand{\Veq}{V_\mathrm{eq}}
\newcommand{\Vl}{V_\mathrm{LWR}}
\newcommand{\vb}{v_{\mathrm{ARZ}}}
\newcommand{\vpsi}{\hat{v}_{\psi}}
\newcommand{\drho}{\delta\hat{\rho}_\beta}
\newcommand{\dv}{\delta\hat{v}_\beta}
\newcommand{\rhoa}{\hat{\rho}_{\alpha}}
\newcommand{\rhoas}{\hat{\rho}_{\alpha^*}}
\newcommand{\rhob}{\rho_{\mathrm{ARZ}}}
\newcommand{\rhol}{\rho_\mathrm{LWR}}
\newcommand{\rhopsi}{\hat{\rho}_{\psi}}

\begin{document}

\title{Second-Order Physics-Informed Learning of \\Road Traffic State using Probe Vehicles
\thanks{\textsuperscript{1} Department of Decision and Control Systems, KTH Royal Institute of Technology, Stockholm, Sweden (e-mail: \{sarabg, barreau\}@kth.se).} 
\thanks{\textsuperscript{2} Université Paris Dauphine-PSL, Paris, France (e-mail: hadrien.nouvel@dauphine.eu).} 
\thanks{This work was partially supported by the Wallenberg AI, Autonomous Systems and Software Program (WASP) funded by the Knut and Alice Wallenberg Foundation. }}
\author{Sara Betancur Giraldo\textsuperscript{1}, Hadrien Nouvel\textsuperscript{2}, Matthieu Barreau\textsuperscript{1}}

\maketitle

\begin{abstract}
We propose a Physics-Informed Learning \mbox{framework} for reconstructing traffic density and velocity state from sparse trajectory data. 
The proposed framework relies on the simple idea that a second-order model captures non-steady-state behaviors as minor modifications to first-order solutions. 
Therefore, the approach implements a tailored second-order correction to the density and the speed of a first-order estimate. 
This framework is evaluated in a ring-road scenario with real and simulated data. 
Results show that the implemented framework improves state estimation by achieving lower error and variance than previously investigated frameworks. 
These results hold for scenarios with heterogeneous driver behaviors (real data) and with different driver behavior models (simulated data), consistent across different penetration rates. 
\end{abstract}

\section{Introduction}
Traffic congestion has negative environmental and public health consequences, such as increased pollution, longer commute times, and greater accident risk~\cite{ferrara2018freeway}. 
Congestion is the result of high road traffic density, which is the state when there is a high number of cars in traffic and, as a consequence, braking or lane changes can grow into stop-and-go waves~\cite{kerner2000}. 
Control strategies to reduce congestion, such as ramp metering and variable speed limits~\cite{Papageorgiou2003}, require reliable information on the current traffic state to be effective.

The current state of the art in traffic state estimation techniques ranges from \mbox{model-based} to data-driven methods~\cite{seo2017traffic}. 
Model-based methods fit a fundamental diagram and conservation laws to measurements~\cite{DAGANZO1994}, so their accuracy depends on how well the model and its calibrated parameters describe the actual traffic. 
In contrast, data-driven methods learn the relationship between measurements and the traffic state directly from data~\cite{seo2017traffic}.
This avoids model mismatch but requires large, representative data sets, and the estimates are not guaranteed to be physically consistent or transferable to another road segment.
Both methods are therefore limited by sparse measurements: model-based methods fill the gaps with a model whose assumptions may not hold, and data-driven methods with patterns that may not apply. 
A hybrid approach, such as \gls{abk:pil}, integrates the benefits of both methods~\cite{Karniadakis2021}. 

The emergence of autonomous vehicles has introduced new opportunities to measure traffic density. 
Autonomous vehicles can act as probe vehicles: vehicles that record their position, speed, and distance to the vehicle ahead while driving. 
One of the main advantages of these probe vehicles is the ability to increase measurement coverage along their trajectories and avoid reliance on fixed sensors~\cite{seo2017traffic}. 
However, probe measurements remain sparse, and the traffic state between their trajectories must still be inferred. 

In this context, traffic flow models play a central role in reconstructing current states. 
These models describe \mbox{either} individual vehicles, as in microscopic models~\cite{Treiber2014}, or \mbox{aggregated} quantities such as density and velocity, as in macroscopic models~\cite{Lighthill1955,Richards1956}. 
Probe vehicles provide microscopic measurements, whereas control often requires the macroscopic state. 
This motivates a coupled \mbox{micro-macro} framework~\cite{barreau2021physics}. 
Among macroscopic models, first-order \mbox{models} assume that the velocity is a function exclusively of the density, and therefore cannot represent out-of-equilibrium states. 
Second-order models include velocity as a state with its own dynamics, capturing nonequilibrium behavior such as stop-and-go waves~\cite{aw2000, Zhang2002}. 
However, these models require additional parameters that are hard to calibrate and train~\cite{shi2021physics}. 

To address these calibration issues in model-based estimation techniques, and the need for well-spread measurements in data-driven techniques, \gls{abk:pil} methods are proposed as a bridge. 
Previous studies on these methods have reconstructed traffic density from probe vehicles using first-order models~\cite{liu2020learning, barreau2021physics} and neural operators~\cite{harting2025closed}. 
Other works consider second-order traffic flow models with data from loop detectors supported by point-wise measurements well spread across the domain~\cite{shi2021physics}, and with data from boundary measurements~\cite{Zhao2022}. 
In this work, we reconstruct a second-order traffic state solely from sparse probe vehicle trajectories, a setting in which we observe that training the second-order model from scratch, as in~\cite{shi2021physics}, is unstable. 
Our main idea is to learn a first-order solution and correct it with a second-order model. 
We validate this method on real ring-road experiments with stop-and-go waves~\cite{stern2018dissipation} and with simulated data on different driver behavior models. 

The contributions of this paper are two-fold:
\begin{itemize}
    \item We develop a nested learning framework that corrects a first-order \gls{abk:pil} reconstruction with a second-order model, using only probe vehicle trajectories. 
    \item We conduct an evaluation on real ring-road experiments and on simulations with three driver models and compare our \gls{abk:pil} framework with three other frameworks. 
\end{itemize}

The remainder of this paper is structured as follows. 
In Section~\ref{sec: background}, we introduce three different physics-based traffic models and their respective couplings. 
In Section~\ref{sec: methodology}, we detail our proposed learning framework to integrate \mbox{model-based} and data-driven methods within a \gls{abk:pil} approach. 
In Section~\ref{sec: results}, we describe our case study and evaluate the \gls{abk:pil} reconstruction of traffic density and velocity in four different learning frameworks. 
Finally, Section~\ref{sec:conclusion} draws conclusions and discusses directions for future work. 

\paragraph*{Notation} For a differentiable function $u:\mathbb{R}^2\to\mathbb{R}$, $(t, x) \mapsto u(t,x)$, we use $\partial_t u$ and $\partial_x u$ to define its derivative with respect to the first (second, respectively) variable. We denote by $\hat{\phi}$ a learned quantity, and by $\hat{u}_\theta$ the neural network approximation of $ u$ with parameters $\theta$ on a compact domain. 
We define the Mean Squared Error (MSE) over a set ~$\mathcal{P}$ of $M$ evaluation points as $\mathrm{MSE}_{\mathcal{P}}(f) := \frac{1}{M}\sum_{y_i\in\mathcal{P}} |f(y_i)|^2$.

\section{Preliminaries and Problem Statement} \label{sec: background} 
In this section, we investigate traffic dynamics, which can be described at different levels of abstraction, ranging from microscopic models that track individual vehicles to macroscopic models that describe aggregated traffic behavior. 
We integrate both perspectives to enable a data-driven framework for reconstructing traffic density and velocity. 

\subsection{Microscopic Model} 
\subsubsection{Follow-the-leader Model}
We consider a first-order follow-the-leader dynamical system in which each vehicle adapts its velocity based on the distance to the vehicle ahead 
\begin{align}
    \label{eq:ftl}
	\begin{cases}
		\dot{x}_i(t) &= V_{\mathrm{eq}}\left(\frac{1}{x_{i+1}(t) - x_{i}(t)}\right) \quad i \in \{1, ..., N-1\}, \\
		\dot{x}_N(t) &= V_\mathrm{lead}(t).
	\end{cases}
\end{align}
The vehicles start from initial positions $x_1(0) < \dots < x_N(0)$ following the leader vehicle $N$ when $V_\mathrm{lead}\geq0$~\cite{Treiber2014}. 
This model captures individual driver behavior through a simple integrator. 

\subsection{Macroscopic Models} 
Macroscopic models describe traffic flow with continuous variables such as density $\rho(t,x)$ and velocity $v(t,x)$. 
The density is the number of vehicles in an infinitesimally short road at position $x$ and time $t$, normalized by the jam density, so that $\rho \in [0,1]$.
The velocity is the mean speed of the vehicles at $(t,x)$. 
These macroscopic models vary in complexity and suitability depending on the traffic scenario. 

\subsubsection{First-order Model LWR}\label{sec: LWRmodel}
The \gls{abk:lwr} model consists of traffic dynamics based on a conservation law~\cite{Lighthill1955, Richards1956} as 
\begin{align} 
	\label{eq: LWR}
		\partial_t\rho + \partial_x\big( \rho V_{\mathrm{eq}}(\rho) \big) &= 0,
\end{align}
where $\Veq$ is a decreasing equilibrium speed with $\Veq(\rho)=0$ for $\rho\ge 1$. 
Standard choices to parametrize this function include Greenshields~\cite{greenshields1935} or Newell--Daganzo~\cite{newell1993simplified, DAGANZO1994}. 
In this work, we use the Newell--Daganzo model: 
\begin{align}\label{eq:Veq}
    \Veq(\rho; \rho_c) := \min \left\{ \Vmax, W(\rho_c) \frac{1-\rho}{\rho} \right\}. 
\end{align}
The function \mbox{$W(\rho_c) = \Vmax \frac{\rho_c}{1-\rho_c}$} represents the speed at which congestion waves travel backwards. 
In congestion, each driver brakes after the one ahead, this perturbation travels backwards even though vehicles move forward. 
In free flow, this perturbation travels forward with traffic and does not affect the vehicles behind. 
The parameter $\rho_c$ delimits these regimes: free flow (\mbox{$\rho<\rho_c$}) and congestion ($\rho>\rho_c$). 

\begin{remark}
    Note that weak solutions of~\eqref{eq: LWR} are not unique in general, even given some constraints on $V_{\mathrm{eq}}$ \cite[Thm.~1]{barreau2021physics}. 
    Among the possible ones, the physically relevant one is the limit as $\gamma \to 0$ of the solutions of the viscous equation \mbox{$\partial_t\rho + \partial_x\big(\rho\Veq(\rho; \rho_c)\big) = \gamma^2 \partial_{xx}\rho$}.
\end{remark}

While simple and computationally efficient, this model assumes that the speed adapts instantaneously to the density. 
Since $\rho$ may vary abruptly, the corresponding equilibrium velocity $V_{\mathrm{eq}}(\rho; \rho_c)$ can exhibit instantaneous changes, which is physically unrealistic. To address this limitation, the velocity can be modeled to relax towards the equilibrium velocity rather than adjusting instantaneously.
One such formulation is the second-order \gls{abk:arz} model proposed in \cite{aw2000, Zhang2002}, which incorporates an additional evolution equation for velocity. 

\subsubsection{Second-order Model \gls{abk:arz}}\label{sec: ARZmodel}
The \gls{abk:arz} second-order dynamics are described by
\begin{equation}    
    \label{eq: ARZ}
    \left\{
    \begin{array}{l}
        \displaystyle \partial_t\rho + \partial_x(\rho v) = 0, \\
        \displaystyle\partial_t(\rho \omega) + \partial_x(\rho v \omega) = \frac{\rho}{\tau}(V_\mathrm{eq}(\rho; \rho_c)-v), \\
        \displaystyle \omega = v+p(\rho), 
    \end{array}
    \right.
\end{equation} 
where $\omega$ represents the driver's desired velocity and is \mbox{adjusted} by the \mbox{traffic} pressure $p$ that models how drivers adjust their speed as congestion increases. 
Several choices of $p$ exist in the literature~\cite{aw2000}. 
We choose \mbox{$p(\rho) = \Vmax - \Veq(\rho;\rho_c)$}, where $\Vmax$ is the speed limit. 

The model in~\eqref{eq: ARZ} is studied in its nonhomogeneous form, where $\tau > 0$ is the relaxation time and $\Veq(\rho; \rho_c)$ represents the equilibrium velocity. 
In the second equation, the model describes how drivers gradually adjust their speed $v$ towards the equilibrium speed $\Veq(\rho; \rho_c)$ over a time $\tau$ and this adaptation propagates with the flow of traffic. 
This allows the model to represent both propagation and dissipation of congestion waves more accurately than first-order models. 

\begin{remark}
We choose the pressure function $p$ to ensure consistency between the \gls{abk:lwr} and the \gls{abk:arz} model. 
This choice means that when the \gls{abk:arz} model is at equilibrium, $v = \Veq(\rho; \rho_c)$ and disturbances travel either at the vehicle speed or at the characteristic speed. 
These speeds characterize how small disturbances propagate and are described by the eigenvalues of~\eqref{eq: ARZ}: $\lambda_1 = v - \rho p'(\rho) = v + \rho\Veq'(\rho; \rho_c)$ and $\lambda_2 = v$~\cite{aw2000}. 
Moreover, let $g \coloneqq v - \Veq(\rho; \rho_c) = \omega - \Vmax$ be the gap between the speed and its equilibrium value. 
Given this function $g$, the momentum equation in~\eqref{eq: ARZ} reduces to $\dot g =g_t+vg_x = -g/\tau$ when following a vehicle. 
This means that the gap described by $g$ decays exponentially with time \mbox{constant $\tau$}. 
This time constant is the time drivers take to adapt their speed, and $\tau \to 0$ recovers the \gls{abk:lwr} model. 
\end{remark}

\subsection{Coupled Micro-macro Model}\label{sec:coupled}
In~\cite{barreau2021physics}, the authors introduce a coupling between \mbox{microscopic} and macroscopic first-order formulations. 
Therein, a key condition to construct a model that combines~\eqref{eq:ftl} and~\eqref{eq: LWR} is that the equilibrium velocity $V_\mathrm{eq}(\rho; \rho_c)$ is twice differentiable and strictly decreasing. 
Our choice for $V_\mathrm{eq}(\rho;\rho_c)$ makes the function not differentiable at $\rho_c$, but a smooth approximation satisfies these conditions~\cite{barreau2021physics}. 

We employ a first-order follow-the-leader formulation, for which it is well established that, as the number of vehicles increases, the microscopic dynamics converge to the \gls{abk:lwr} macroscopic traffic model \cite{holden2018follow}. 
In particular, vehicle trajectories can be interpreted as samples of the macroscopic velocity field. 

In summary, the first-order model is: 
\begin{equation} 
	\label{eq:coupling_first_order}
	\left\{
    \begin{array}{l}
        \displaystyle \partial_t\rho(t,x) + \partial_x\big(\rho(t,x) V_{\mathrm{eq}}(\rho(t,x); \rho_c)\big) = 0, \\
	   	\displaystyle \dot{x}_i(t) = V_{\mathrm{eq}}(\rho(t, x_i); \rho_c), \quad i \in \{1, \dots, N\},
    \end{array}
    \right.
\end{equation}
and the second-order model is: 
\begin{align} 
	\label{eq:coupling}
    \left\{
    \begin{array}{l}
        \displaystyle \partial_t\rho(t,x) + \partial_x(\rho(t,x) v(t,x)) = 0, \\
        \displaystyle\partial_t(\rho(t,x) \omega(t,x)) + \partial_x(\rho(t,x) v(t,x) \omega(t,x)) = \\
        \qquad\qquad\qquad\frac{\rho(t,x)}{\tau}(V_\mathrm{eq}(\rho(t,x); \rho_c)-v(t,x)), \\
    	\displaystyle\dot{x}_i(t) = v(t, x_i(t)), \quad i \in \{1, \dots, N\}
    \end{array}
    \right.
\end{align}
under the same conditions as in~\cite{barreau2021physics}. 
These formulations enable the integration of trajectory measurements into a second-order macroscopic traffic model, effectively coupling microscopic observations with macroscopic dynamics. 

The resulting second-order system is nonlinear and infinite-dimensional, making both its solution and the design of suitable observers particularly challenging. 
However, developing observers that reliably capture speed adaptation with high-fidelity is necessary for real-life control. 
Previous studies that considered this problem usually implemented \gls{abk:pil} to deal with the heterogeneity between the micro and macro systems and built solutions directly from a formulation similar to~\eqref{eq:coupling}. 
Since learning second-order dynamics from scratch is difficult, the training suffers from stochastic issues, leading to variable performance.

\subsection{Problem Formulation}
Given the coupled system formulation~\eqref{eq:coupling}, we aim to estimate the macroscopic traffic states: density $\rho$ and velocity $v$ over a time interval $t \in [0,T]$ on a ring-road $x \in [0, L)$.

To achieve the estimation objective, we assume \mbox{access} to measurements from probe vehicles described in the set \mbox{$\mathcal{D} := \{(t_k,\, x_i(t_k),\, \rho_i(t_k),\, v_i(t_k)) \}_{k=1,\,i=1}^{K,\,N_\mathcal{D}}$}, where $x_i(t_k)$ is the position of probe vehicle $i$ at time $t_k$ and $\rho_i(t_k) := \rho(t_k, x_i(t_k)),\, v_i(t_k) := v(t_k, x_i(t_k))$ are the local measurements of the density and the speed for the same probe. 
This results in partial observation of the space, and thus coverage depends on the fraction of vehicles that are collect measurements (penetration rate). 

In this paper, we assess how a tailored neural network architecture affects traffic state estimation using second-order traffic models and evaluate its performance relative to first-order models and data-only estimates. 
In particular, we want to compare the quality of the estimates and evaluate whether the framework is reliable across different initializations.

\section{Methodology} \label{sec: methodology} 
We propose a \gls{abk:pil} framework for reconstructing \mbox{macroscopic} traffic density and velocity from sparse \mbox{trajectory} measurements. 
We aim to estimate the traffic state using a data-driven framework and enforce consistency with physical dynamics via~\eqref {eq:coupling_first_order} and~\eqref{eq:coupling}. 

\begin{figure*}[t]
    \centering
    \includegraphics[width=\textwidth]{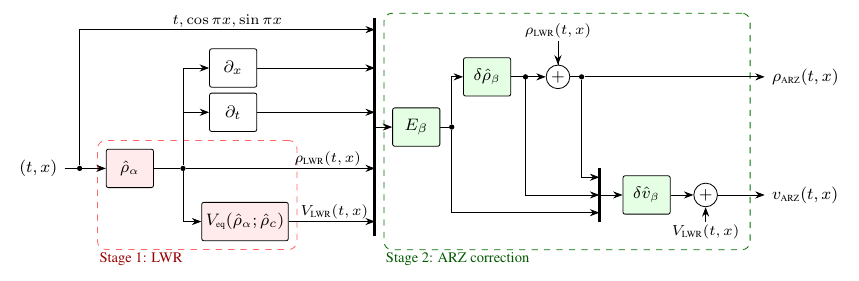}
    \caption{Two-stage learning. Stage~1 shows the first-order Physics-Informed Learning framework implementation as in Section~\ref{sec:training_1} (red). The estimations of Stage~1 are then frozen. Stage~2 shows the second-order Physics-Informed Learning framework implementation as in Section~\ref{sec:training_2n}. Stage~2 receives the point in space, the frozen estimations of the first-order solution at that point and the Taylor expansion of the first-order solution density estimation. In Stage~2, the encoder $E_\beta$ and the correction heads $\delta\rho$, $\delta v$ are trained (green). Stage~2 returns the corrected states $\rho_\mathrm{ARZ}$ and $v_\mathrm{ARZ}$}
    \label{fig:diagram}
\end{figure*}

\subsection{Learning Framework}
We describe three different learning frameworks. 
First, we estimate a density field and equilibrium velocity using a \mbox{first-order} framework. 
Second, we estimate a density and a velocity field employing a second-order framework with nonequilibrium velocities. 
Third, we detail our proposed methodology: 
\begin{itemize} 
    \item We extend the first-order framework by reconstructing a density and a velocity field enabled by a correction with a second-order framework. 
    \item We adjust the parameters of our estimated functions by minimizing a data-fitting cost. 
    \item We include physics constraints from first-order and second-order models in the cost function to guide the solution. 
\end{itemize}
This leads to an optimization problem that combines model-based and data-driven learning, particularly useful under sparse observations.

We define a dataset where the optimization framework evaluates the physics constraints. 
This set of collocation physics points is described by $\mathcal{C} := \{(t_j, x_j)\}_{j=1}^{N_c}$. 
The points are drawn randomly over $[0, T] \times [0, L)$ and redrawn at every training step. 

\subsubsection{First-order training}\label{sec:training_1}

We first estimate the traffic state density and equilibrium velocity using a first-order formulation. 
We approximate the density by a multilayer perceptron (MLP) $\rhoa$ with weights $\alpha$. 
We define its input $(t,x)$ from trajectory observations. 
This approximation returns $\rhoa(t,x) \in (0,1)$. 
The critical density $\hat\rho_c$ and the viscosity coefficient $\hat\gamma$ are scalar parameters, learned jointly with $\alpha$.

We formulate the optimization problem to learn the parameters $\alpha$, including the constraints directly in the optimization \mbox{objective} as a sum of losses as 
\begin{equation} \label{opt:LWR}
    \alpha^*, \hat\rho_c^* \in \argmin_{\alpha, \hat\rho_c} \min_{\hat\gamma} \mathcal{L}_{\rho} + \mathcal{L}_{V} + \mathcal{L}_{v,\rho} + \mathcal{L}_{\mathrm{LWR}} + \mathcal{L}_{\gamma}, 
\end{equation}
where
\begin{equation*}
    \begin{aligned}
        \mathcal{L}_{\rho} & := \mathrm{MSE}_\mathcal{D}(\rhoa(t_k,x_i) - \rho_i(t_k)), \\
        \mathcal{L}_{V} & := \mathrm{MSE}_\mathcal{D} (\Veq(\rho(t_k,x_i); \hat{\rho}_c) - v_i(t_k)), \\
        \mathcal{L}_{v,\rho} & := \mathrm{MSE}_\mathcal{D} (\Veq(\rhoa(t_k,x_i); \hat{\rho}_c) - v_i(t_k)), \\
        \mathcal{L}_{\gamma} & := \hat{\gamma}^2, \\
        \mathcal{L}_{\mathrm{LWR}} & := \mathrm{MSE}_\mathcal{C} (\mathcal{N}_\mathrm{LWR}(t_j,x_j; \alpha, \hat\rho_c,\hat\gamma))
    \end{aligned}
\end{equation*}
In $\mathcal{L}_{\rho}$ and $\mathcal{L}_{V}$, we evaluate the mismatch \mbox{between} estimation and measurements for the density and \mbox{velocity} functions. 
In $\mathcal{L}_{v,\rho}$, we evaluate the \mbox{mismatch} between the estimated velocity as a by-product of the estimated density and the measured velocity. 
In $\mathcal{L}_{\gamma}$, we encode the constraint for the dissipation \mbox{coefficient $\hat\gamma$} to stay small. 
In $\mathcal{L}_{\mathrm{LWR}}$, we evaluate the physics constraints based on the first-order model dynamics in~\eqref{eq:coupling_first_order} according to
\begin{align*}
    \mathcal{N}_\mathrm{LWR}(\cdot; \alpha, \hat\rho_c,\hat\gamma) &:= \partial_t (\rhoa) + \partial_x \big(\rhoa V_{\mathrm{eq}}(\rhoa;\hat{\rho}_c)\big) \\ 
    & \qquad - \hat\gamma^2 \partial_{xx}\rhoa, 
\end{align*}
which guides the \mbox{reconstruction} where no probe data are available.

In~\eqref{opt:LWR}, we encode a first-order model estimation. 
These types of models describe congestion wave propagation well in steady-state scenarios. 
However, they tend to perform poorly in scenarios involving sudden wave-congestion formation. 
In these cases, drivers' speeds can change and behave differently from an equilibrium state and can differ among drivers. 
The following subsection details the second-order training, which is expected to estimate the traffic states better under such scenarios. 

\subsubsection{ARZ-direct training}\label{sec:training_2d}
We estimate the traffic state density and equilibrium velocity using a second-order formulation. 
We approximate the density and velocity by an MLP with inputs $(t,x)$ and outputs $\rhopsi$ and $\vpsi$, with weights $\psi$. 
This approximation returns $\rhopsi(t,x) \in (0,1)$ and $\vpsi(t,x) \in [0, \Vmax]$. 
The critical density $\hat\rho_c$ and relaxation time $\hat\tau$ are scalar parameters, learned jointly with $\psi$.

We formulate the optimization problem to learn the parameters $\psi$ including the constraints directly in the optimization \mbox{objective} as a sum of losses as 
\begin{equation} \label{opt:ARZdirect}
    \psi^* \in \argmin_\psi \min_{\hat\rho_c,\hat\tau} \mathcal{L}_{\rho} + \mathcal{L}_{v} + \mathcal{L}_{V,v} + \mathcal{L}_{\mathrm{mass}} + \mathcal{L}_{\mathrm{mom}}, 
\end{equation}
where
\begin{equation*}
    \begin{aligned}
        \mathcal{L}_{\rho} & := \mathrm{MSE}_\mathcal{D}(\rhopsi(t_k,x_i) - \rho_i(t_k)), \\
        \mathcal{L}_{v} & := \mathrm{MSE}_\mathcal{D} (\vpsi(t_k,x_i) - v_i(t_k)), \\
        \mathcal{L}_{V,v} & := \mathrm{MSE}_\mathcal{D} (\Veq(\rho_i(t_k); \hat{\rho}_c) - \vpsi(t_k,x_i)), \\
        \mathcal{L}_{\mathrm{mass}} & := \mathrm{MSE}_\mathcal{C} (\mathcal{N}_\mathrm{mass}(t_j,x_j; \psi)), \\
        \mathcal{L}_{\mathrm{mom}} & := \mathrm{MSE}_\mathcal{C} (\mathcal{N}_\mathrm{mom}(t_j,x_j; \psi, \hat\rho_c,\hat\tau)).
    \end{aligned}
\end{equation*}
In $\mathcal{L}_{\rho}$ and $\mathcal{L}_{v}$, we quantify the mismatch between the estimated and measured density and velocity functions. 
In $\mathcal{L}_{V,v}$, we fit the equilibrium velocity to the measured density-speed pairs of the probe vehicles, which identifies $\hat\rho_c$ from the data. 
In $\mathcal{L}_{\mathrm{mass}}$ and $\mathcal{L}_{\mathrm{mom}}$, we evaluate the physics constraints according to the ARZ model dynamics as~\eqref{eq:coupling}. 
The mass conservation penalty is:
\[
    \mathcal{N}_\mathrm{mass}(\cdot; \psi) := \partial_t \rhopsi + \partial_x \big(\rhopsi \vpsi\big),
\]
and the momentum conservation penalty is:
\begin{multline}
    \mathcal{N}_\mathrm{mom}(\cdot; \psi, \hat\rho_c,\hat\tau) := (\rhopsi \omega)_t + (\rhopsi \vpsi \omega)_x \\
    - \frac{\rhopsi}{\hat{\tau}}(V_\mathrm{eq}(\rhopsi; \hat{\rho}_c)-\vpsi).
\end{multline}
Both constrain the reconstruction when no probe data are available and propagate information between measurements. 
The redundancy between the measurements should thereby enable the identification of $\hat{\tau}$. 
A difficulty with this training is that it attempts to solve a coupled differential equation with unknown parameters, which is challenging, especially in the given low-data regime~\cite{shi2021physics}. 
In the following section, we propose a nested training as a solution to this challenge. 

\begin{figure*}[t]
	\centering
	\includegraphics[width=0.99\textwidth]{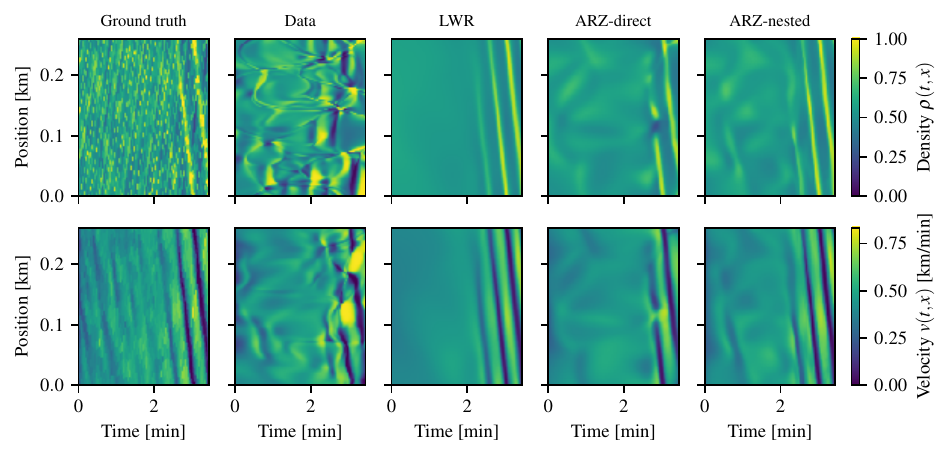}
	\caption{Density (top) and velocity (bottom) fields for the real data experiments with $4$ probe vehicles~\cite{stern2018dissipation}. The spatiotemporal fields (from left to right) are from ground truth (first) and estimations by the learning frameworks: data (second), LWR or first-order (third), ARZ-direct (fourth), and ARZ-nested (fifth). The experiments show congestion waves propagation, columns $2-5$ are their estimations. }
	\label{fig:spatiotemporal}
\end{figure*} 

\subsubsection{ARZ-nested training}\label{sec:training_2n}
Previous studies have shown that the training of a first-order solution is efficient and robust \cite{barreau2021physics}. 
This could enable a stabilization to the training of the more accurate second-order model by rooting its solution on a correction to the first-order one. 
This second-order formulation learns therefore a minor modification to the solution, thereby simplifying training and limiting the potential for poor extrapolation. 
We now explain the proposed architecture, shown in Fig.~\ref{fig:diagram}. 

First, we gather the density and equilibrium velocity at the required position as the output of the LWR training. 
We leverage the knowledge of $\alpha^*$ and $\hat\rho_c^*$ from the first training stage. 
However, the first-order model outputs equilibrium velocities; therefore, it cannot be extended solely from point values. 
Hence, we extend the point value by adding the derivatives of the density at this point, allowing us to extrapolate the neighborhood of each point using a Taylor expansion. 
We construct the vector \mbox{$z = [t, \cos(\pi x), \sin(\pi x), \rhol, \partial_x\rhol, \partial_t\rhol, \Vl]^\top$} as an input for the second stage. 

The second stage starts with an encoder network which, from $z$, creates a latent-space representation containing the needed information to reproduce the behavior of the first-order model in a neighborhood of $(t,x)$. 
This representation is the input to an MLP $\delta \hat\rho_\beta$, which aims to compute the difference between the second-order density and the first-order density. 
Similarly, we construct a vector with the representation, the corrected density, and the local difference in density. 
This vector is the input to an MLP $\delta \hat v_\beta$, which computes the difference between the equilibrium velocity and the true car velocity. 
The heads $\drho(t,x)$ and $\dv(t,x)$, therefore, correct the frozen LWR solution by including better representations of the individual driver behavior in the solution. 
We define the corrected states as
\begin{equation*}
\begin{aligned}
  \rhob &:= \sigma\big(\mathrm{logit}(\rhol) + \drho\big), \\ 
  \vb &:= \Vmax\,\sigma\big(\mathrm{logit}(\Vl/\Vmax) + \dv\big), 
  \end{aligned}
\end{equation*}
where $\rhol=\rhoas$ and \mbox{$\Vl=\Veq(\rhoas; \hat{\rho}_c^*)$}. 
This approximation returns $\rhob(t,x)$ and \mbox{$\vb(t,x)$}. 
The relaxation time $\hat\tau$ is a scalar parameter, learned jointly with $\beta$. 

\begin{remark}
    By construction, $\rhob\in(0,1)$ and \mbox{$\vb\in(0,\Vmax]$} for any weights $\beta$, and $\drho=\dv=0$ recover the LWR solution. 
    Note that when $\rhol\le\rho_c$, $\Vl=\Vmax$ and hence $\vb=\Vmax$ for any $\dv$. 
This means the speed is corrected only where the first-order solution is congested. 
\end{remark}

\begin{remark}
    This tailored neural architecture is designed to represent a stage-wise learning in which a refinement of the solution is proposed in the second stage, extending the idea proposed in \cite{eshkofti2025vanishing}, where a similar but simpler strategy is employed over a parameter in an equation. 

    This specific choice decouples the learning of $V_{\mathrm{eq}}$ and $\hat\tau$, which we believe is one reason why training both together is unstable with sparse data. 
    Moreover, building on a first-order solution, which is easier to estimate even with sparse measurements, enables learning the second-order solution with a limited number of probe vehicles.
\end{remark}

We formulate the optimization problem to learn the parameters $\beta$ including the constraints directly in the optimization \mbox{objective} as a sum of losses as  
\begin{equation} \label{opt:ARZnested}
    \beta^* \in \argmin_\beta \min_{\hat\tau} \mathcal{L}_{\rho} + \mathcal{L}_{v} + \mathcal{L}_{\mathrm{mass}} + \mathcal{L}_{\mathrm{mom}}, 
\end{equation}
where
\begin{equation*}
    \begin{aligned}
        \mathcal{L}_{\rho} & := \mathrm{MSE}_\mathcal{D}(\rhob(t_k,x_i) - \rho_i(t_k)), \\
        \mathcal{L}_{v} & := \mathrm{MSE}_\mathcal{D} (\vb(t_k,x_i) - v_i(t_k)), \\
        \mathcal{L}_{\mathrm{mass}} & := \mathrm{MSE}_\mathcal{C} (\mathcal{N}_\mathrm{mass}(t_j,x_j; \beta)), \\
        \mathcal{L}_{\mathrm{mom}} & := \mathrm{MSE}_\mathcal{C} (\mathcal{N}_\mathrm{mom}(t_j,x_j; \beta, \hat\rho_c^*,\hat\tau)).
    \end{aligned}
\end{equation*}
The losses in~\eqref{opt:ARZnested} follow the same relationships as in~\eqref{opt:ARZdirect}, detailed in Section~\ref{sec:training_2d}. 


To compare the effectiveness of each model at estimating density and velocity in traffic, we evaluate their performance across different datasets. 
The results of these experiments are detailed in Section~\ref{sec: results}.

\section{Results}\label{sec: results}
We evaluate the traffic density and velocity state reconstructions using real and simulated data on a ring-road scenario. 
We detail the real-life ring-road experiments from Stern et al.~\cite{stern2018dissipation} in Section~\ref{sec:results_setup}. 
Section~\ref{sec:results_discusion} discuses the reconstruction performance on the real data for each of the frameworks and different penetration rates. 
Section~\ref{sec:results_comparison} compares the density reconstruction trained on simulated and real data for each of the proposed frameworks. 

\subsection{Experimental setup} \label{sec:results_setup}
\subsubsection*{Preprocessing}
We consider data from the ring-road experiments of Stern et al.~\cite{stern2018dissipation}. 
In these experiments, a $260\mathrm{m}$ ring with $21-22$ vehicles driven by human drivers record their positions and speeds. 
The experiments include one autonomous vehicle acting as a control actuator to dissolve congestion waves. 
We use five datasets that describe the intervals in the experiments when the control was inactive. 

The positions from the experimental data allow us to compute the density of each vehicle as the jam spacing divided by the distance to the vehicle ahead. 
The jam spacing, which describes when vehicles are stopped bumper-to-bumper ($\rho=1$), is $7.03 \mathrm{m}$ and was obtained as $0.95$ times the median spacing of stopped cars. 
Additionally, we obtain the maximum speed $\Vmax= 0.834 \mathrm{km/min}$. 

\begin{figure}[t]
	\centering
	\includegraphics[width=\linewidth]{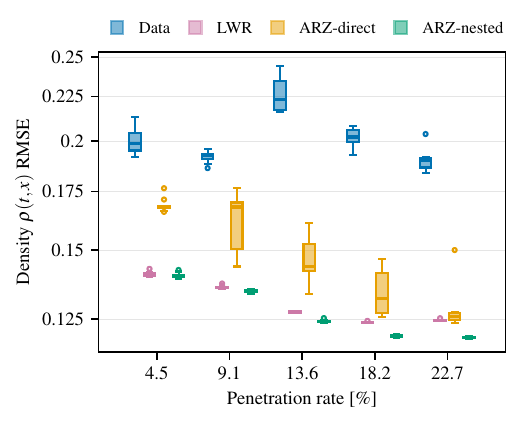}
	\includegraphics[width=\linewidth]{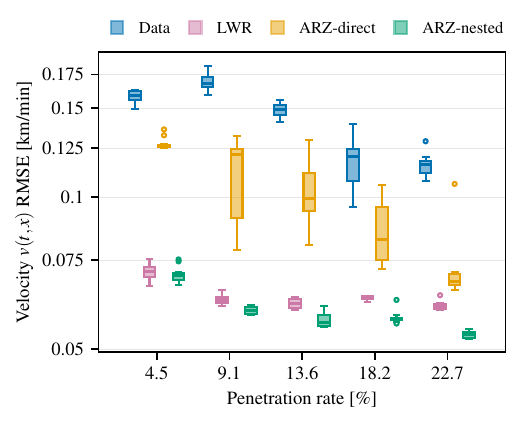}
	\caption{Distribution of $\mathrm{RMSE}$ on $\rho(t,x)$ (top) and $v(t,x)$ (bottom) over $10$ seeds for different penetration rates ($1-5$ vehicles, $4.5\%-23\%$). We compare the variability of the errors achieved by the different learning frameworks proposed: data, LWR or first-order, ARZ-direct and ARZ-nested. We use a logarithmic scale on the y-axis. }
	\label{fig:error}
\end{figure} 

\subsubsection*{Training}
We evaluate the estimation of the density and velocity for the different frameworks described in Section~\ref{sec: methodology}: data, LWR, ARZ-direct, and ARZ-nested. 

We train each framework for $13,000$ epochs. 
Specifically, for the ARZ-nested model, we initialize with LWR training for $10,000$ epochs and finalize with ARZ training ($3,000$ additional epochs). 
In all these training runs, we use the Adam optimization algorithm with a learning rate of $1\times10^{-3}$ and cosine decay to $1\times10^{-4}$. 
Additionally, we employ $1,000$ collocation points per step to evaluate our physics constraints. 

We train each framework with $10$ different random seeds on the real data and $20$ on the simulated data to study the stability and variability of our learning frameworks. 
The seed changes the initial weights and the collocation points, not the selected probe vehicles. 
These trainings consisted of reconstructions for each described interval and for different penetration rates ($1-5$ vehicles, $4.5\%-23\%$). 

The resulting estimations for the density and velocity fields are presented and discussed in Section~\ref{sec:results_discusion} below. 

\subsection{Reconstruction} \label{sec:results_discusion}


In Fig.~\ref{fig:spatiotemporal}, the ground truth shows stop-and-go waves traveling backward along the ring for both the density and the velocity. 
The data framework reconstructs the traffic state with patterns that do not propagate as waves. 
Additionally, as shown in Fig.~\ref{fig:error}, this framework exhibits the largest error at every rate and does not decrease steadily with more probes. 
The LWR framework captures the strongest waves but smooths out the weaker waves across time and space. 
In this framework, the algebraic coupling between the density and velocity, enforced by $V_\mathrm{eq}(\cdot;\rho_c)$, is evident and explains the results. 
The ARZ-direct framework yields a reconstruction similar to LWR but enables estimation of more details in the smoothed-out parts of LWR. 
This may be due to the nonequilibrium representation of the speed state, which fits the speed measurements more accurately. 
However, this framework has a large variance across runs, for this specific penetration rate, the recovery is accurate and robust, but for other penetration rates, this result does not hold. 
As observed in \cite{shi2021physics}, when the measurements within the domain are not sparse, the reconstruction accuracy and variability improve drastically.    
The ARZ-nested framework retains all the features of the LWR reconstruction and slightly better recovers the weaker waves in the ground truth. 
This means that our ARZ-nested framework reconstructs traffic phenomena before a wave starts, enabling the use of this estimate for control. 
As expected, since it is much easier to learn a correction than the whole field from scratch, the ARZ-nested approach also benefits from the low variability of the LWR framework.
The ARZ-nested framework reduces the error of LWR by $2.3\,\%$ in density and $7.2\,\%$ in velocity, on average across penetration rates. 

The probe vehicles leave most of the evaluation space unobserved; therefore, the data framework is expected to perform worse than other frameworks. 
For these gaps, the physics-based losses constrain the reconstruction with patterns that are physically meaningful. 
This explains the improvement in the errors of the \gls{abk:pil} frameworks. 
In the ARZ-nested framework, this improvement could be attributed to the independent velocity state the second-order model adds. 

In summary, the application of physics knowledge to data-driven estimation improves performance, reducing the error by regularizing the estimation and recovering physically meaningful phenomena. 
It is worth noting that real data is harder to reconstruct because it captures heterogeneous driving behavior. 
Our implementation of the nested learning framework captures these behaviors and is supported by data, while proving robust across different seeds.

\begin{table}[t]
\vspace{5pt}
\centering
\caption{Median RMSE on the simulated data, $20$ probe vehicles ($4.0\,\%$ penetration), $20$ seeds. Best in bold.}
\label{tab:sim}
\footnotesize
\setlength{\tabcolsep}{1.7pt}
\begin{tabular}{ll cccc c cccc}
 & & \multicolumn{4}{c}{Density RMSE} & & \multicolumn{4}{c}{Velocity RMSE [km/min]} \\
\cmidrule(lr){2-6}\cmidrule(lr){7-11}
 & & Data & LWR & Direct & Nested & & Data & LWR & Direct & Nested \\
\midrule
Krauss   & & 0.029 & 0.027 & 0.027 & \textbf{0.026} & & 0.030 & 0.030 & 0.030 & \textbf{0.028} \\
IDM      & & 0.089 & 0.086 & 0.086 & \textbf{0.080} & & 0.059 & 0.062 & 0.060 & \textbf{0.054} \\
IDM mix  & & 0.098 & 0.089 & 0.091 & \textbf{0.084} & & 0.057 & 0.061 & 0.060 & \textbf{0.053} \\
\end{tabular}
\end{table}

\subsection{Robustness across datasets} \label{sec:results_comparison}

We train our models on data simulated with \gls{abk:sumo}~\cite{sumo,SUMOsimulator}. 
We simulate a $6.28$\,km ring with $502$ vehicles, with the same mean density and speed limit as the real experiments. 
We use three driver models: Krauss, IDM~\cite{Treiber2014}, and IDM with $20\%$ cautious drivers. 
In this setup, we use $5$ to $35$ probe vehicles (with $1\%-7\%$ penetration rates) and $20$ seeds. 
We aim to include a larger-scale scenario and different driver behaviors to test whether the conclusions hold beyond one experiment. 

Tables~\ref{tab:sim} and~\ref{tab:real} report the median errors at a $4\%$ penetration rate for the simulated data and about $14\%$ for the real data. 
We choose to show these penetration rates as they fall in the middle of each tested range. 
The ARZ-nested framework has the lowest density error across all datasets and the lowest velocity error, except on dataset A2 from the real data, where LWR is slightly better.  
Its improvement over LWR is larger in velocity (up to $13\%$) than in density ($1$ to $7\%$). 
The ARZ-direct framework has higher error than the ARZ-nested framework in most datasets. 
This is consistent across most penetration rates. 
However, in simulated scenarios, the data framework achieves a very similar error rate at higher penetration rates. 

These results confirm our claims on the real data scenario discussed in Section~\ref{sec:results_discusion}. 
The extra velocity state in ARZ-nested could explain the $7\%$ improvement in the velocity estimation in the Krauss driver behavior model and the $13\%$ in the IDM models. 
The data framework's performance in simulated scenarios could be explained by the characteristics of the generated waves. 
These waves are longer than the gaps between the probe vehicles; therefore, simple interpolation can easily reconstruct them. 

In summary, our ARZ-nested framework provides a good solution and stable estimates without adding complexity to the reconstruction while using the same $13,000$ budget.

\begin{table}[t]
\vspace{5pt}
\centering
\caption{Median RMSE on the real data~\cite{stern2018dissipation}, $3$ probe vehicles ($13.6-14.3\,\%$ penetration), $10$ seeds. Best in bold.}
\label{tab:real}
\footnotesize
\setlength{\tabcolsep}{1.7pt}
\begin{tabular}{ll cccc c cccc}
 & & \multicolumn{4}{c}{Density RMSE} & & \multicolumn{4}{c}{Velocity RMSE [km/min]} \\
\cmidrule(lr){2-6}\cmidrule(lr){7-11}
 & & Data & LWR & Direct & Nested & & Data & LWR & Direct & Nested \\
\midrule
A1       & & 0.196 & 0.163 & 0.166 & \textbf{0.159} & & 0.154 & 0.066 & 0.069 & \textbf{0.061} \\
A2       & & 0.200 & 0.162 & 0.162 & \textbf{0.159} & & 0.093 & \textbf{0.047} & 0.048 & 0.051 \\
B1       & & 0.211 & 0.158 & 0.164 & \textbf{0.156} & & 0.157 & 0.068 & 0.072 & \textbf{0.062} \\
B2       & & 0.192 & 0.148 & 0.151 & \textbf{0.143} & & 0.163 & 0.070 & 0.075 & \textbf{0.061} \\
C1       & & 0.223 & 0.127 & 0.143 & \textbf{0.124} & & 0.149 & 0.061 & 0.099 & \textbf{0.056} \\
\end{tabular}
\end{table}

\section{Conclusion} \label{sec:conclusion}
In this paper, we investigated \gls{abk:pil} frameworks for reconstructing traffic density and velocity states on a ring-road under different data capturing driver behaviors. 
We developed a nested estimation framework that incorporates a second-order correction to a frozen first-order estimate. 
This novel framework enables density and velocity estimation guided by sparse trajectory data while respecting physical traffic laws. 

The results highlight that implementing a second-order correction improves reconstruction accuracy, particularly for the velocity state. 
This improvement holds true even for real data with heterogeneous driver behaviors. 
However, in this paper, we proposed an offline reconstruction that may not be suitable for control. 
Additionally, the ring-road scenario differs from real traffic infrastructure, which could include open roads and more lanes. 

Future work could extend this reconstruction framework to larger traffic networks, where sparse measurements can be used to estimate traffic states in real time for \mbox{multiple} road segments. 
Such applications could complement \mbox{existing} work on transportation networks, including the use of \mbox{\gls{abk:cavs}} to enhance population accessibility, as explored in~\cite{SALAZAR2024}.

\section*{Acknowledgment}

Computation resources were provided by the National Academic Infrastructure for Supercomputing in Sweden (NAISS), partially funded by the Swedish Research Council through grant agreement no. 2022-06725.

\section*{Use of AI Statement}
During the preparation of this work, the authors used Claude to generate the case studies and code, run the experiments, and improve phrasing. After using this tool, the authors reviewed, edited, and verified the content as needed and take full responsibility for the final content.

\bibliographystyle{ieeetr}
\bibliography{biblio}

\end{document}